\documentclass[a4paper,UKenglish,cleveref,hyperref,autoref]{lipics-v2019}
\title{4 vs 7 sparse undirected unweighted Diameter is SETH-hard at time $n^{4/3}$}
\titlerunning{4 vs 7 sparse undirected unweighted Diameter is SETH-hard at time $n^{4/3}$}
\author{\'{E}douard Bonnet}{Univ Lyon, CNRS, ENS de Lyon, Université Claude Bernard Lyon 1, LIP UMR5668, France}{edouard.bonnet@ens-lyon.fr}{https://orcid.org/0000-0002-1653-5822}{}

\authorrunning{\'E. Bonnet}

\Copyright{Édouard Bonnet}%TODO mandatory, please use full first names. LIPIcs license is "CC-BY";  http://creativecommons.org/licenses/by/3.0/

%TODO mandatory: Please choose ACM 2012 classifications from https://dl.acm.org/ccs/ccs_flat.cfm 

%TODO mandatory; please add comma-separated list of keywords

\category{}%optional, e.g. invited paper

\supplement{}%optional, e.g. related research data, source code, ... hosted on a repository like zenodo, figshare, GitHub, ...

\nolinenumbers %uncomment to disable line numbering

\hideLIPIcs  %uncomment to remove references to LIPIcs series (logo, DOI, ...), e.g. when preparing a pre-final version to be uploaded to arXiv or another public repository

\EventEditors{John Q. Open and Joan R. Access}
\EventNoEds{2}
\EventLongTitle{42nd Conference on Very Important Topics (CVIT 2016)}
\EventShortTitle{CVIT 2016}
\EventAcronym{CVIT}
\EventYear{2016}
\EventDate{December 24--27, 2016}
\EventLocation{Little Whinging, United Kingdom}
\EventLogo{}
\SeriesVolume{42}
\ArticleNo{23}
\usepackage{xspace}
\usepackage[table]{xcolor}
\usepackage{tikz}
\usepackage{array}
\usepackage{mathdots}
\usepackage{pdfpages}

\usetikzlibrary{fit}
\usetikzlibrary{calc}
\usetikzlibrary{shapes,patterns}
\usetikzlibrary{decorations.pathreplacing}
\usepackage{algpseudocode,algorithm,algorithmicx}

\algrenewcommand\algorithmicrequire{\textbf{Precondition:}}
\algrenewcommand\algorithmicensure{\textbf{Postcondition:}}
\usepackage{caption}
\usepackage{subcaption}
\usepackage{verbatim}
\usepackage{todonotes}

\newcommand{\fov}{\textsc{4-OV}\xspace}
\newcommand{\kov}{\textsc{$k$-OV}\xspace}
\newcommand{\ltov}{\textsc{2-Orthogonal Vectors}\xspace}
\newcommand{\lfov}{\textsc{4-Orthogonal Vectors}\xspace}
\newcommand{\lkov}{\textsc{$k$-Orthogonal Vectors}\xspace}

\newcommand{\dist}{d}
\newcommand{\diam}{\text{diam}}

\newcommand{\ind}{\text{ind}}
\newcommand{\maj}{\text{maj}}

\theoremstyle{plain}

\newtheorem{conjecture}[theorem]{Conjecture}
\newtheorem{observation}[theorem]{Observation}

\definecolor{g1}{rgb}{0,0,1}
\definecolor{g2}{rgb}{0.2,0,0.8}
\definecolor{g3}{rgb}{0.4,0,0.6}
\definecolor{g4}{rgb}{0.6,0,0.4}
\definecolor{g5}{rgb}{0.8,0,0.2}
\definecolor{g6}{rgb}{1,0,0}

\begin{document}

\maketitle

\begin{abstract}
  We show, assuming the Strong Exponential Time Hypothesis, that for every $\varepsilon > 0$, approximating undirected unweighted \textsc{Diameter} on $n$-vertex $m$-edge graphs within ratio $7/4 - \varepsilon$ requires $m^{4/3 - o(1)}$ time, even when $m = \Tilde{O}(n)$.
  This is the first result that conditionally rules out a near-linear time $5/3$-approximation for undirected \textsc{Diameter}. 
\end{abstract}

\section{Introduction}\label{sec:intro}

The diameter of a graph is the length of a longest shortest path between two of its vertices.
We write \textsc{Diameter} for the algorithmic task of computing the diameter of an input graph.
Throughout the paper, $n$ implicitly denotes the number of vertices of a graph, and $m$, its number of edges.
We will often prefix \textsc{Diameter} with \emph{undirected/directed} to indicate whether or not edges may be oriented\footnote{In directed~\textsc{Diameter}, we are to compute the length of a longest shortest path taken from any vertex to any vertex.}, and \emph{unweighted/weighted} to indicate whether or not non-negative edge weights are allowed. 

A fairly recent and active line of work aims to determine the best runtime for an algorithm approximating \textsc{Diameter} within a given ratio.
First, there is an exact algorithm running in time\footnote{where $\Tilde{O}(\cdot)$ suppresses the polylogarithmic factors} $\Tilde{O}(mn)$, which computes $n$ shortest-path trees from every vertex of the graph.
Secondly, there is a $2$-approximation running in time $\Tilde{O}(m)$, which computes a shortest-path tree from an arbitrary vertex and outputs the largest distance found.
There are an $\Tilde{O}(m^{3/2})$ time $3/2$-approximation for directed weighted \textsc{Diameter}~\cite{Aingworth99,Roditty13,Chechik14}, and for every non-negative integer $k$, an $\Tilde{O}(m n^{\frac{1}{k+1}})$ time $(2-2^{-k})$-approximation\footnote{with an extra additive factor depending on the weights} for \emph{undirected} weighted \textsc{Diameter}~\cite{Cairo16}.
We refer the interested reader to the survey of Rubinstein and Vassilevska Williams \cite{Rubinstein19}.

We will now focus on sparse graphs, for which $m = \Tilde{O}(n)$.
This is because the current paper deals with conditional lower bounds on approximating~\textsc{Diameter}, and all such results even work with that restriction.
Observe that, on sparse graphs, the first result of the previous paragraph is a near-quadratic 1-approximation, while the second result is a near-linear \mbox{2-approximation}.
One can represent these ratio-runtime trade-offs in the two-dimensional plane.
The ultimate goal of fine-grained complexity, in that particular context, is to obtain a complete curve of algorithms linking these two extreme points, matched by tight conditional lower bounds.
We now present one way of deriving conditional lower bounds for polytime problems.

\paragraph*{Lower bounds based on the Strong Exponential Time Hypothesis}

The Strong Exponential Time Hypothesis (SETH, for short) asserts that for every $\varepsilon > 0$, there is an integer $k$ such that \textsc{$k$-SAT} cannot be solved in time $(2-\varepsilon)^n$ on $n$-variable instances~\cite{Impagliazzo01}.
At first glance, this assumption should only be useful to rule out some specific running time for NP-hard problems which, like the satisfiability problem, seems to require superpolynomial time.
Such conditional lower bounds to classical~\cite{Cygan16} or parameterized algorithms~\cite{Cygan15} are overviewed in a survey~\cite{Lokshtanov11} on the consequences of the SETH (as well as the weaker assumption ETH) on solving computationally hard problems.

Interestingly, using the SETH to rule out a given running time for a polynomial-time solvable problem took more time.
In a survey of fine-grained complexity~\cite{WilliamsSurvey}, Vassilevska Williams dates the first reduction (albeit used positively) from \textsc{SAT} to a problem in P back to 2005~\cite{Williams05}.
We will see that this reduction to \ltov, where one wants to find two orthogonal $0,1$-vectors within a given list, is very relevant to the fine-grained complexity of~\textsc{Diameter}.
As it turns out, the first SETH-based lower bound for a polytime graph problem occurred almost a decade later, on the very unweighted undirected \textsc{Diameter}~\cite{Roditty13}.

There might have been a psychological barrier in reducing a ``hard'' problem to an ``easy'' one, in order to derive a conditional lower bound.
However this makes perfect sense.
Let us give an apropos example.
Suppose (as it is actually the case) that one can create in time $O(n)$ a~list of $n$ 0,1-vectors with $n = O(2^{N/2})$, from an $N$-variable \textsc{SAT} formula, such that there is pair of orthogonal vectors in the list if and only if the formula is satisfiable.
Now a truly subquadratic algorithm, that is in time $n^{2-\varepsilon}$ for some $\varepsilon > 0$, for \ltov would enable to solve \textsc{SAT} in time $O(2^{(1-\varepsilon/2)N})=O((2-\delta)^N)$ for some $\delta > 0$, contradicting the SETH.
We thus say that \ltov is \emph{SETH-hard} at time~$n^2$, and more generally a problem $\Pi$ is \emph{SETH-hard} at time $T$ if it requires time $T^{1-o(1)}$ under the SETH.  

\paragraph*{SETH lower bounds for \textsc{Diameter}}

There is a handful of SETH-hardness results on approximating \textsc{Diameter}~\cite{Roditty13,Backurs18,Bonnet20,Wein20,Li20,Li20b}.
Unless the SETH fails, any $3/2 - \varepsilon$-approximation for sparse undirected unweighted~\textsc{Diameter}, with $\varepsilon > 0$, requires time $n^{2 - o(1)}$~\cite{Roditty13} (this is the above-mentioned seminal result to the fine-grained complexity within P), whereas any $5/3 - \varepsilon$-approximation requires time $n^{3/2 - o(1)}$~\cite{Li20} (an early version of~\cite{Li20b}).
Since a $5/3$-approximation of \textsc{Diameter} running in near-linear time was consistent with the then knowledge (up until mid-August 2020, even in weighted directed graphs) Rubinstein and Vassilevska Williams \cite{Rubinstein19} and Li~\cite{Li20} ask for such an algorithm or some lower bounds with a ratio closer to~2.

In the last few months, there were several developments on directed graphs. 
The author showed that, under the SETH, $7/4 - \varepsilon$-approximating sparse directed weighted~\textsc{Diameter} requires time $n^{4/3 - o(1)}$~\cite{Bonnet20}.
Then Wein and Dalirrooyfard~\cite{Wein20}, and independently, Li~\cite{Li20b} (an updated version of~\cite{Li20}) both show that not only this result holds on directed \emph{unweighted} graphs but they generalize it in the following way: Unless the SETH fails, for every $\varepsilon > 0$ and every integer $k \geqslant 4$, $\frac{2k-1}{k}-\varepsilon$-approximating directed unweighted~\textsc{Diameter} requires time $n^{\frac{k}{k-1} - o(1)}$.

Despite these advances, a near-linear time $5/3$-approximation for the \emph{undirected}~\textsc{Diameter} may still have existed. 
In this paper, we rule out this possibility by showing the following (see~\cref{fig:results} for a visual summary of what is now known on approximating undirected~\textsc{Diameter}).

\begin{theorem}\label{thm:main}
  Unless the SETH fails, for any $\varepsilon > 0$, $7/4 - \varepsilon$-approximating \textsc{Diameter} on undirected unweighted $n$-vertex $\Tilde{O}(n)$-edge graphs requires $n^{4/3 - o(1)}$ time. 
\end{theorem}

In particular we resolve \cite[Open Question 2.2.]{Rubinstein19}, on the existence of a near-linear time $5/3$-approximation for \emph{undirected}~\textsc{Diameter}, by the negative.

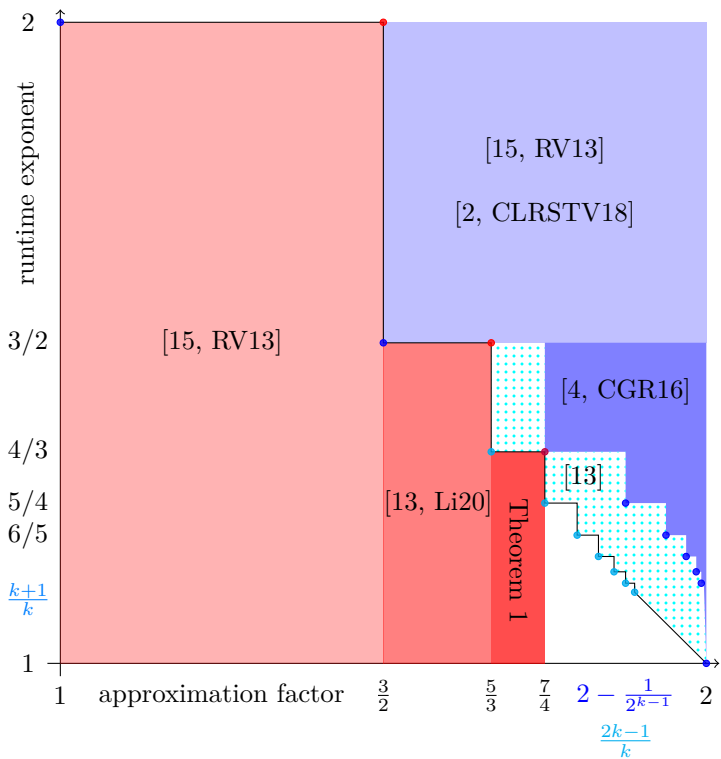
\begin{figure}[h!]
  \centering
  \begin{tikzpicture}
    \def\s{8.5}
    \def\xb{1}
    \def\xe{2}
    \def\yb{1}
    \def\ye{2}
    \def\h{0.02}
    %scale
    \draw[->] (\xb * \s, \yb * \s - \h * \s) --  (\xb * \s, \ye * \s + \h * \s) ;
    \draw[->] (\xb * \s - \h * \s, \yb * \s) --  (\xe * \s + \h * \s, \yb * \s) ;
    \foreach \i/\j in {1/1,1.1/{\textcolor{blue!50!cyan}{$\frac{k+1}{k}$}},1.2/{6/5},1.25/{5/4},1.33/{4/3},1.5/{3/2},2/2}{
      \node at (\xb * \s - 0.05 * \s, \i * \s) {\j} ;
    }
    \foreach \i/\j in {1/1,1.5/{$\frac{3}{2}$},1.667/{$\frac{5}{3}$},1.75/{$\frac{7}{4}$},1.875/{\textcolor{blue}{$2-\frac{1}{2^{k-1}}$}},2/2}{
      \node at (\i * \s, \yb * \s - 0.05 * \s) {\j} ;
    }
    \node at (1.875 * \s, \yb * \s - 0.12 * \s) {\textcolor{cyan}{$\frac{2k-1}{k}$}} ;

    \node at (1.25 * \s, \yb * \s - 0.05 * \s) {approximation factor} ;
    \node at (\xb * \s - 0.05 * \s, 1.75 * \s) {\rotatebox{90}{runtime exponent}} ;
    %SETH lower bound
    \fill[red,opacity=0.3] (\xb * \s, \yb * \s) -- (1.5 * \s, \yb * \s) -- (1.5 * \s, \ye * \s) -- (\xb * \s, \ye * \s) -- cycle ;
    \node at (1.25 * \s, 1.5 * \s) {\cite[RV13]{Roditty13}} ;

    \fill[red,opacity=0.5] (1.5 * \s, \yb * \s) -- (1.667 * \s, \yb * \s) -- (1.667 * \s, 1.5 * \s) -- (1.5 * \s, 1.5 * \s) -- cycle ;
    \node at (1.585 * \s, 1.25 * \s) {\cite[Li20]{Li20b}} ;

    \fill[red,opacity=0.7] (1.667 * \s, \yb * \s) -- (1.75 * \s, \yb * \s) -- (1.75 * \s, 1.33 * \s) -- (1.667 * \s, 1.33 * \s) -- cycle ;
    \node at (1.71 * \s, 1.16 * \s) {\rotatebox{-90}{\cref{thm:main}}} ;

    %algorithms
    \fill[blue,opacity=0.25] (1.5 * \s, 1.5 * \s) -- (1.5 * \s, \ye * \s) -- (\xe * \s, \ye * \s) -- (\xe * \s, 1.5 * \s) -- cycle ;
    \node at (1.75 * \s, 1.8 * \s) {\cite[RV13]{Roditty13}} ;
    \node at (1.75 * \s, 1.7 * \s) {\cite[CLRSTV18]{Backurs18}} ;
    
    \fill[blue,opacity=0.5] (1.75 * \s, 1.5 * \s) -- (1.75 * \s, 1.33 * \s) -- (1.875 * \s, 1.33 * \s) -- (1.875 * \s, 1.25 * \s) -- (1.9375 * \s, 1.25 * \s) -- (1.9375 * \s, 1.2 * \s) -- (1.96875 * \s, 1.2 * \s) -- (1.96875 * \s, 1.1667 * \s) -- (1.984375 * \s, 1.1667 * \s) -- (1.984375 * \s, 1.1428 * \s) -- (1.9921875 * \s, 1.1428 * \s) -- (1.9921875 * \s, 1.125 * \s) -- (1.99609375 * \s, 1.125 * \s) -- (\xe * \s, \yb * \s) -- (\xe * \s, 1.5 * \s) -- cycle ;
    \node at (1.875 * \s, 1.425 * \s) {\cite[CGR16]{Cairo16}} ;

    %SETH lower bounds barriers
     \begin{scope}
       \pgfsetfillpattern{dots}{cyan}
       \fill (1.667 * \s, 1.33 * \s) -- (1.75 * \s, 1.33 * \s) -- (1.75 * \s, 1.5 * \s) -- (1.667 * \s, 1.5 * \s) -- cycle ;
     \fill (1.889 * \s, 1.111 * \s) -- (1.889 * \s, 1.125 * \s) -- (1.875 * \s, 1.125 * \s) -- (1.875 * \s, 1.1428 * \s) -- (1.8571 * \s, 1.1428 * \s) -- (1.8571 * \s, 1.1667 * \s) -- (1.833 * \s, 1.1667 * \s) -- (1.833 * \s, 1.2 * \s) -- (1.8 * \s, 1.2 * \s) -- (1.8 * \s, 1.25 * \s) -- (1.75 * \s, 1.25 * \s) -- (1.75 * \s, 1.33 * \s) -- (1.875 * \s, 1.33 * \s) -- (1.875 * \s, 1.25 * \s) -- (1.9375 * \s, 1.25 * \s) -- (1.9375 * \s, 1.2 * \s) -- (1.96875 * \s, 1.2 * \s) -- (1.96875 * \s, 1.1667 * \s) -- (1.984375 * \s, 1.1667 * \s) -- (1.984375 * \s, 1.1428 * \s) -- (1.9921875 * \s, 1.1428 * \s) -- (1.9921875 * \s, 1.125 * \s) -- (1.99609375 * \s, 1.125 * \s) -- (\xe * \s, \yb * \s) -- cycle;
     \end{scope}
    \draw (\xe * \s, \yb * \s) -- (1.889 * \s, 1.111 * \s) -- (1.889 * \s, 1.125 * \s) -- (1.875 * \s, 1.125 * \s) -- (1.875 * \s, 1.1428 * \s) -- (1.8571 * \s, 1.1428 * \s) -- (1.8571 * \s, 1.1667 * \s) -- (1.833 * \s, 1.1667 * \s) -- (1.833 * \s, 1.2 * \s) -- (1.8 * \s, 1.2 * \s) -- (1.8 * \s, 1.25 * \s) -- (1.75 * \s, 1.25 * \s) -- (1.75 * \s, 1.33 * \s) -- (1.667 * \s, 1.33 * \s) -- (1.667 * \s, 1.5 * \s) -- (1.5 * \s, 1.5 * \s) -- (1.5 * \s, \ye * \s) -- (\xb * \s, \ye * \s) ;

    \node at (1.81 * \s, 1.29 * \s) {\cite{Li20b}} ;
    %dots
    \foreach \i/\j/\c in {1/2/blue,1.5/2/red,1.5/1.5/blue,1.667/1.5/red,1.75/1.33/purple,1.875/1.25/blue,1.9375/1.2/blue,1.96875/1.1667/blue,1.984275/1.1428/blue,1.9921875/1.125/blue,2/1/blue, 1.667/1.33/cyan, 1.75/1.25/cyan, 1.8/1.2/cyan, 1.833/1.1667/cyan,1.8571/1.1428/cyan,1.875/1.125/cyan,1.889/1.111/cyan}{
      \draw[\c] (\i * \s, \j * \s) circle [radius=1.2pt] ;
      \fill[\c,opacity=0.8] (\i * \s, \j * \s) circle [radius=1.2pt] ;
    }
  \end{tikzpicture}
  \caption{Approximability of sparse undirected unweighted~\textsc{Diameter}.
  Blue areas are feasible, as witnessed by algorithms at bottom-left corners (blue dots).
  The red regions are SETH-hard, as witnessed by reductions at top-right corners (red dots).
  Dotted cyan areas are not SETH-hard, unless the NSETH fails.
  The current landscape for the sparse undirected \emph{weighted}~\textsc{Diameter} is the same, except the middle red region is due to~Backurs et al.~\cite{Backurs18} instead of~\cite{Li20b}.
  The axis-parallel black curve represents the tractability frontier as foreseen by~\cref{conj:main}.}
  \label{fig:results}
\end{figure}

In light of the recent results (see in particular the paragraph on barriers to SETH-hardness), it is reasonable to conjecture that the four variants of sparse~\textsc{Diameter} (undirected/directed unweighted/weighted) are equally approximable.
More precisely, we venture the following optimistic prediction.

\begin{conjecture}\label{conj:main}
  Sparse (un)directed (un)weighted~\textsc{Diameter} is $2 - \frac{1}{k}$-approximable in time $\tilde{O}(n^{\frac{k+1}{k}})$ for every $k \in \mathbb N^+ \cup \{\infty\}$.
  Unless the SETH fails, approximating sparse (un)directed (un)weighted~\textsc{Diameter} within ratio better than $2 - \frac{1}{k+1}$ requires time $n^{\frac{k+1}{k}-o(1)}$ for every $k \in \mathbb N^+$.
\end{conjecture}

\cref{conj:main} is naturally equivalent to obtaining the algorithms for the directed weighted \textsc{Diameter} and the SETH-hardness for the undirected unweighted~\textsc{Diameter}.
Settling the conjecture would give a complete landscape of the approximability of~\textsc{Diameter}, where if one represents the results in the two-dimensional space of approximation factor vs runtime exponent, the feasible and infeasible regions are separated by a rectilinear curve with infinitely many corners (the black curve drawn in~\cref{fig:results}).
In that respect, our contribution is to give the third lower bound on the curve (i.e., North-East corner) after Roditty and Vassilevska Williams gave the first~\cite{Roditty13}, and Li, the second~\cite{Li20}.
Hopefully our new ideas (together with the recent constructions in the directed case of Wein and Dalirrooyfard~\cite{Wein20}, and Li~\cite{Li20b}) will also help in generalizing the lower bound predicted by~\cref{conj:main} to every positive integer~$k$.

\paragraph*{Barriers to SETH lower bounds}

\cref{conj:main} is partly prompted by intriguing results due to Li~\cite{Li20b}.
To state them, we need to recall the definition of a strengthening of SETH introduced by Carmosino et al.~\cite{Carmosino16}.
It is called NSETH for Nondeterministic SETH.
NSETH asserts that for every $\varepsilon > 0$, there is an integer $k$ such that the \textsc{$k$-Taut} problem cannot be solved in \emph{non-deterministic} time $(2-\varepsilon)^n$, where \textsc{$k$-Taut} asks, given a $k$-DNF formula whether every truth assignment satisfies it (in other words, if it is a tautology). 
Li shows, for all four variants of~\textsc{Diameter} \emph{but} the directed weighted one, that no point positioned strictly above the rectilinear black curve of~\cref{fig:results} can be shown SETH-hard, under the NSETH (and, if randomized reductions are permitted, under a stronger assumption, called NUNSETH for Non-Uniform NSETH).

\cref{conj:main} is very optimistic since it predicts that every such point will be explained by an algorithm.
There are many alternatives to that event.
For instance NSETH could be false\footnote{If we are totally honest, even the weaker SETH does not gather such a wide consensus, and is false if quantum computation is allowed.}, or the intractability region could extend further North via a non SETH-based reduction, or via a deterministic SETH-based reduction in the directed weighted case. 
Besides it would require significant progress in approximating the sparse directed~\textsc{Diameter}, when currently no algorithm running in time $n^{\frac{3}{2}-\varepsilon}$ achieves approximation factor better than 2.

The second half of~\cref{conj:main} shown for every $k \geqslant 4$ on directed graphs~\cite{Wein20,Li20b}, and for $k=1,2,3$ on undirected graphs, is much easier to believe in.

\paragraph*{Techniques}

Like every mentioned \textsc{Diameter} lower bound (for more details, see the paragraphs on \kov and \textsc{Diameter} in the surveys~\cite{WilliamsSurvey,Rubinstein19}), we reduce from \lkov, where one seeks, in a given set of $N$ $0,1$-vectors of dimension $\ell$, $k$ vectors such that at every index, at least one of these $k$ vectors has a 0 entry. 
Under the SETH, \lkov requires time $N^{k-o(1)}$~\cite{Williams05}, even when $\ell$ is polylogarithmic in $N$.

Here we will reduce from \lfov.
We thus wish to build a graph on $\Tilde{O}(N^3)$ vertices and edges with diameter 7 if there is an orthogonal quadruple (i.e., a solution to the the \lfov instance), and diameter 4 otherwise.
Following a reduction to $ST$-\textsc{Diameter}\footnote{where one seeks the length of a longest shortest path from a vertex of $S$ to a vertex of $T$} by Backurs et al.~\cite{Backurs18} (arguably also following~\cite{Roditty13}) most of the reductions (as in \cite{Bonnet20,Wein20,Li20b}) feature layers $L_0, L_1, \ldots, L_{k-1}, L_k$, with only (forward) edges between two consecutive $L_i$. 
The vertices within the same layer share the same number of ``vector attributes'' and ``index attributes''.
The interplay between vector and index attributes in defining the vertices and edges is made so that if there are no $k$~orthogonal vectors, then there are paths of ``optimal'' length $k$ between every pair in $L_0 \times L_k$, whereas if there is set~$X$ of $k$~orthogonal vectors, a pair in $L_0 \times L_k$ jointly encoding~$X$ is suddenly very far apart (usually and ideally at distance $2k-1$).

Then the challenge is to make sure that, on NO-instances, the other pairs (not in $L_0 \times L_k$) are at distance at most $k$, without destroying the previous property.
The core of our reduction is similar to our previous construction for directed weighted~\textsc{Diameter}~\cite{Bonnet20}.
However we simplify and streamline it in the following way.
As in the first construction of Li~\cite{Li20}, we collapse some layers into one.
We will have $L_0 = L_4 (= T)$ and $L_1 = L_3 (= C)$, while $L_2$ is called $P$.
This makes the case analyses simpler (fewer kinds of pairs to consider).

At this point, we face the same issue as in~\cite{Bonnet20}: There are pairs in $T \times P$ that are too far apart.
On directed graphs, this can be fixed by adding parallel layers and appropriate ``back'' edges \cite{Bonnet20,Wein20} or simply ``back'' edges \cite{Li20b}. 
This is no longer an option.
Instead we add a set $I$ of vertices with only index attributes.
These vertices link the right pairs of $T \times P$ with path of length~4 (we are back to the first variation on the theme~\cite{Roditty13}).
To emphasize that the situation is somewhat delicate, we observe that not all the pairs of $T \times P$ can be at distance~4, since otherwise every pair in $T \times T$ is at distance at most~6. 
We set $I$ at distance 3 of $T$ (by initially putting edges of weight~3).
This permits to cliquify $I$ without creating $TT$-paths of length at most~6.
In turn, this puts every pair involving $I$ at distance at most~4, as well as pairs of $(C \cup P) \times P$.
Note that as long as $\dist(T,X)+\dist(T,Y) \geqslant 3$ (or $k-1$), one can have \emph{all} the pairs of $X \times Y$ at distance~4 (or $k$), without creating undesired $TT$-paths of length at most~6 (or $2k-2$).

We then remove the weight-3 edges between $T$ and $I$.
This involves some vertex splits transforming $T$ into $T, T', T''$, and a simpler echo of the idea of having the clique $I$, with a clique $I'$ connecting appropriately the pairs in $T \times T''$.  

\paragraph*{Organization}

In~\cref{sec:prelim}, we recall graph-theoretic notations, and give the relevant background on the \textsc{Orthogonal Vectors} problem.
In~\cref{sec:weighted}, we present a simpler reduction with edge weights.
It thus achieves the statement of~\cref{thm:main} for sparse undirected weighted~\textsc{Diameter}. 
In~\cref{sec:unweighted}, we tune this reduction to get rid of the edge weights, and establish~\cref{thm:main}. 

\section{Preliminaries}\label{sec:prelim}

We use standard graph-theoretic notations.
If $G$ is a graph, $V(G)$ denotes its vertex set, and $E(G)$, its edge set.
We denote the edge set between $X \subseteq V(G)$ and $Y \subseteq V(G)$ by $E(X,Y)$.
If $S \subseteq V(G)$, $G[S]$ denotes the subgraph of $G$ induced by $S$. %, and $G - S$ is a short-hand for $G[V(G) \setminus S]$.
\emph{Weighted graphs} have positive edge weights.
(Throughout the paper, we will only need edges of weight~1 and 3.) 
We exclusively deal with undirected graphs (for which the distance function is symmetric). 
For $u, v \in V(G)$, $\dist_G(u,v)$ denotes the distance between $u$ and $v$ in $G$, that is, the number of edges in a shortest path between $u$ and $v$. 
For every positive integer $r$ and every vertex $u \in V(G)$, $N^r_G[u]$ denotes the set of vertices $v$ such that $d_G(u,v) \leqslant r$. 
In unweighted graphs, the closed neighborhood of $u$, denoted $N_G[u]$, coincides with $N^1_G[u]$.
However in a weighted graph $N^1_G[u]$ would for instance not contain the neighbors of $u$ via an edge of weight greater than~1.
This subtlety will arise only once, and we will remind the reader in due time.
For every positive integer $r$ and every vertex $S \subseteq V(G)$, $N^r_G[S]$ denotes the set of vertices $v \in V(G)$ such that $d_G(u,v) \leqslant r$ for some $u \in S$.
We observe that, in unweighted graphs, $N^r_G[S]$ coincides with $N_G[N_G[\cdots N_G[N_G[S]] \cdots ]]$ where $N_G[\cdot]$ is applied $r$~times and $N_G[S]$ is the closed neighborhood of $S$.
We drop the subscript in the above notations, if the graph $G$ is clear from the context.

We denote by $\diam(G)$ the diameter of $G$, that is, $\max_{u,v \in V(G)} \dist(u,v)$.
The~\textsc{Diameter} problem asks, given a graph $G$, for the value of $\diam(G)$.
We call $uv$-path, a path going from vertex $u$ to vertex $v$, and $ST$-path (with possibly $S=T$), any path going from some vertex $u \in S$ to some vertex $v \in T$. 

If $\ell$ is a positive integer, $[\ell]$ denotes the set $\{1,2,\ldots,\ell\}$.
If $v$ is a vector and $i$ is a positive integer, then $v[i]$ denotes the $i$-th coordinate of $v$.
We use $\maj(a_1,\ldots,a_h)$ to denote the value with the largest number of occurrences in the tuple $(a_1,\ldots,a_h)$.

For every fixed positive integer~$k$, the \lkov (\kov for short) problem is as follows.
It asks, given a set~$S$ of 0,1-vectors in $\{0,1\}^\ell$, if there are~$k$ vectors $v_1, \ldots, v_k \in S$ such that for every $i \in [\ell]$, $\Pi_{h \in [k]} v_h[i] = 0$, or equivalently, $v_1[i] = v_2[i] = \cdots = v_k[i] = 1$ does not hold.
Williams~\cite{Williams05} showed that, assuming the SETH, \kov requires $N^{k-o(1)}$ time with $N := |S|$.
Furthermore, using the Sparsification Lemma~\cite{Sparsification}, this lower bound holds even when, say, $\ell = \lceil \log^2 N \rceil$.
Here we will leverage this lower bound for $k=4$.
This is, in the context of the SETH-hardness of approximating \textsc{Diameter}, a usual opening step: For example, Roditty and Vassilevska Williams~\cite{Roditty13} uses this lower bound for $k = 2$, Li~\cite{Li20}, for $k=3$ and general $k \geqslant 3$, the author~\cite{Bonnet20}, for $k=4$, Wein and Dalirrooyfard~\cite{Wein20}, for general $k \geqslant 5$ and $k=4$.

\section{A simpler reduction with edge weights}\label{sec:weighted}

From any set $S$ of $N$ vectors in $\{0,1\}^\ell$, we build an undirected weighted graph $G = \rho(S)$ (with edge weights 1 and 3, only) with $O(N^3+N^2\ell^3+\ell^5)$ vertices and $O(N^3\ell^5+N^2\ell^6+\ell^{10})$ edges such that if $S$ admits an orthogonal quadruple then the diameter of $G$ is (at least)~7, whereas if $S$ has no orthogonal quadruple then the diameter of $G$ is (at most) 4.
We recall that \fov requires $N^{4-o(1)}$ time, unless the SETH fails, even when $\ell = \lceil \log^2 N \rceil$~\cite{Williams05}.
In that case, the graph $G$ has $O(N^3)$ vertices and $\tilde{O}(N^3)$ edges.
Hence any algorithm approximating sparse undirected weighted \textsc{Diameter} within ratio better than $7/4$ in time $n^{4/3 - \delta}$, with $\delta > 0$, would refute the SETH.

\subsection{Construction}\label{sec:construction}

We first describe the vertex set of $G$, then its edge set, and finally check that the number of vertices and edges are as announced.

\paragraph*{Vertex set}

Every vertex of $G$ is the concatenation of a possibly empty tuple of vectors of $S$, called \emph{vector tuple}, followed by a possibly empty tuple of possibly equal indices of $[\ell]$, called~\emph{index tuple}.
Each coordinate of the vector tuple is called a \emph{vector field}, while each coordinate of the index tuple is called an \emph{index field}.
The set $V(G)$ is partitioned into four sets: $T$ (for \textbf{t}riples), $C$ (for \textbf{c}ouples), $P$ (for \textbf{p}airs), and $I$ (for \textbf{i}ndices).
The names behind $T, C, P$ reflect the number and the nature (ordered or unordered) of the vector fields.
Each of these sets comprise vertices with up to three vector fields and five index fields.
They are defined in the following way.
\begin{itemize}
\item $T$: for every $\{a,b,c\} \in {S \choose 3}$, we add vertex $(a,b,c)$ to $T$.
  Thus vertices of $T$ have three vector fields and no index field.
\item $C$: for every $\{a,b\} \in {S \choose 2}$ and $i,j,k \in [\ell]$ such that $a[i] = a[j] = a[k] = 1$ and $\maj(b[i],b[j],b[k])=1$, we add vertex $(a,b,i,j,k)$ to $C$.
  Thus vertices of $C$ have two vector fields and three index fields.
\item $P$: for every $a,b \in S$ and $i,j,k \in [\ell]$ such that $a[i] = a[j] = a[k] = 1$ and $b[i] = b[j] = b[k] = 1$, we add vertex $(\{a,b\},i,j,k)$ to $P$.
  We will still see $a$ and $b$ as filling the \emph{two} vector fields of the vertex, without a \emph{first} vector field and a \emph{second} vector field.
  Contrary to vertices of $C$, $(\{a,b\},i,j,k)$ and $(\{b,a\},i,j,k)$ are two names for the same vertex (whereas $(a,b,i,j,k)$ and $(b,a,i,j,k)$ are two distinct vertices, whose existence implies slightly different properties).
  Thus vertices of $P$ also have two vector fields and three index fields.
  Note also that $\{a,b\}$ is a multiset, since $a$ may be equal to $b$.
\item $I$: for every $p_1,p_2,i,j,k \in [\ell]$, we add vertex $(p_1,p_2,i,j,k)$ to $I$.
  The chosen labels for the five index fields anticipate that, to build the edge set, it is convenient to imagine a~separation after the first two index fields of the tuple. 
  The vertices of $I$ have no vector field and five index fields.
\end{itemize}

\paragraph*{Edge set}

We will put some edges between $T$ and $C$, $C$ and $P$, $P$ and $I$, and $I$ and $T$.
In addition, we put \emph{index-switching edges} within $I$ and within $C$.
An index-switching edge is between two vertices of the same set ($I$ or $C$) with the same vector tuple (which is always the case in~$I$) and distinct index tuples.
The only edges with a weight different than~1 are the edges between $I$ and $T$, which all have weight~3.
Thus, unless specified otherwise, an edge has weight~1.

The total list of edges is as follows.
\begin{itemize}
\item We add all the index-switching edges within $I$ and $C$.
  Thus $G[I]$ is a clique and $G[C]$ is a disjoint union of at most ${|S| \choose 2}$ cliques (while $G[T]$ and $G[P]$ remain independent sets).
  More explicitly, we have an edge between every pair of distinct vertices $(p_1,p_2,i,j,k) \in I$ and $(p'_1,p'_2,i',j',k') \in I$, and for every $a \neq b \in S$ between every pair of distinct vertices $(a,b,i,j,k) \in C$ and $(a,b,i',j',k') \in C$.  
\item $E(T,C)$: We add an edge between every $(a,b,c) \in T$ and $(a,b,i,j,k) \in C$ provided that there is an $h \in \{i,j,k\}$ such that $b[h] = c[h] = 1$.
\item $E(C,P)$: We add an edge between every $(a,b,i,j,k) \in C$ and $(\{c,d\},i,j,k) \in P$ whenever $a \in \{c,d\}$.
\item $E(T,I)$: We add an edge of weight 3 between every $(a,b,c) \in T$ and $(p_1,p_2,i,j,k) \in I$ whenever $a[p_1] = b[p_1] = c[p_1] = a[p_2] = b[p_2] = c[p_2] = 1$. 
\item $E(I,P)$: We add an edge between every $(p_1,p_2,i,j,k) \in I$ and $(\{a,b\},i,j,k) \in P$ whenever $a[p_1] = b[p_2] = 1$ or $a[p_2] = b[p_1] = 1$.
\end{itemize}

This ends the construction.
See~\cref{fig:weighted} for an illustration.

\begin{figure}[h!]
  \centering
  \begin{tikzpicture}
    \def\h{5}
    \def\v{3}
    \foreach \i/\j/\l/\ll/\x/\y in {0/0/{(a,b,c)}/abc/0/-0.5, 0/1/{(a,b,i,j,k)}/abijk/-1.2/0, 0.3/1/{(a,b,i',j',k')}/abijk2/0/-0.5, 0/2/{(\{d,e\},i,j,k)}/deijk/-1.3/0, -1.2/1/{(p_1,p_2,i,j,k)}/ind/0/0.5,-1.5/1/{(p'_1,p'_2,i',j',k')}/indp/0/-0.5}{
      \node[draw, circle] (\ll) at (\h * \i,\v * \j) {} ;
      \node (t\ll) at (\h * \i + \x,\v * \j + \y) {$\l$} ;
    }

    \foreach \i/\j/\c in {abc/abijk/black, abijk/abijk2/black,abijk/deijk/black,abc/ind/blue,ind/deijk/black,ind/indp/black}{
      \draw[thick, color=\c] (\i) -- (\j) ;
    }
    \node (d1) at (- 0.5 * \h, 0.1 * \v) {} ;
    \node (d2) at (0.5 * \h, -0.1 * \v) {} ;
    \node (d1p) at (- 0.5 * \h, 2.1 * \v) {} ;
    \node (d2p) at (0.5 * \h, 1.9 * \v) {} ;

    \small{}
    \node (ab1) at (- 0.32 * \h, 1.15 * \v) {$\mathbf{a[i]=a[j]=a[k]=1}$} ;
    \node (ab2) at (- 0.35 * \h, 0.85 * \v) {\textbf{maj}$\mathbf{(b[i],b[j],b[k])=1}$} ;
    \node (de1) at (- 0.18 * \h, 2.15 * \v) {$\mathbf{d[i]=d[j]=d[k]=e[i]=e[j]=e[k]=1}$} ;
    \node at (0.03 * \h, 0.55 * \v) {$\mathbf{\exists h \in \{i,j,k\},}$} ;
    \node at (0.03 * \h, 0.35 * \v) {$\mathbf{c[h]=b[h]=1}$} ;
    \node at (0.04 * \h, 1.55 * \v) {$\mathbf{a \in \{d,e\}}$} ;
    \node at (-0.8 * \h, 0.55 * \v) {$\mathbf{a[p_1]=b[p_1]=c[p_1]=}$} ;
    \node at (-0.8 * \h, 0.35 * \v) {$\mathbf{a[p_2]=b[p_2]=c[p_2]=1}$} ;
    \node at (-0.85 * \h, 0.78 * \v) {\textcolor{blue}{3}} ;
    \node at (-0.82 * \h, 1.65 * \v) {$\mathbf{d[p_1]=e[p_2]=1}$} ;
    \node at (-0.82 * \h, 1.55 * \v) {\textbf{or}} ;
    \node at (-0.82 * \h, 1.45 * \v) {$\mathbf{d[p_2]=e[p_1]=1}$} ;
    \normalsize{}

    \node[draw,rectangle,rounded corners,thick,fit=(ab1) (ab2) (tabijk2)] (C) {} ;
    \node[draw,rectangle,rounded corners,thick,fit=(abc) (tabc) (d1) (d2)] (T) {} ;
    \node[draw,rectangle,rounded corners,thick,fit=(de1)  (deijk) (tdeijk) (d1p) (d2p)] (P) {} ;
    \node[draw,rectangle,rounded corners,thick,fit=(tind) (tindp)] (I) {} ;

    \foreach \i/\j/\k in {0.4/0/T, 0.4/1.1/C, 0.4/2/P, -1.7/1.1/I}{
      \node (t\k) at (\i * \h, \j * \v) {$\k$} ;
    }
    \foreach \i/\j/\k in {0.48/0.135/{(3,0)}, 0.48/1.235/{(2,3)}, 0.48/2.235/{(2,3)}, -1.72/1.245/{(0,5)}}{
      \node[red] at (\i * \h, \j * \v) {\tiny{$\k$}} ;
    }
  \end{tikzpicture}
  \caption{The weighted construction $G$.
    In bold, the conditions for the existence of a vertex or of an edge.
    The edge in blue, and more generally every edge of $E(T,I)$, has weight~3, while all other edges have weight~1.
    The pairs in red recall, for vertices of the corresponding set, the length of their vector tuple followed by the length of their index tuple.
  }
  \label{fig:weighted}
\end{figure}
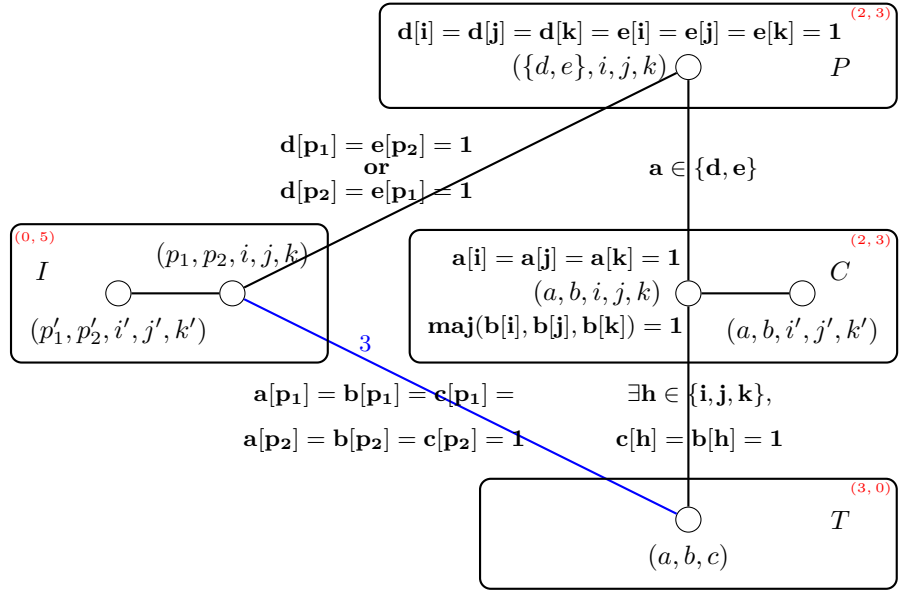

\paragraph*{Vertex and edge count}

There are $O(N^3)$ vertices in $T$, $O(N^2 \ell^3)$, in $C \cup P$, and $\ell^5$, in $I$, hence $O(N^3+N^2\ell^3+\ell^5)=O(N^3)$ in total.
There are $O(N^3 \ell^3)$ edges in $E(T,C) \cup E(C,P)$, $O(N^2 \ell^6)$, in $E(C)$, $O(N^3 \ell^5)$, in $E(T,I)$, $O(N^2 \ell^5)$, in $E(I,P)$, and $O(\ell^{10})$ in $E(I)$, hence $O(N^3 \ell^5+N^2 \ell^6+\ell^{10})=\tilde{O}(N^3)$ edges in total.
Furthermore $G$ can be built in time $\Tilde{O}(N^3)$.

\subsection{The absence of orthogonal quadruple implies diameter at most 4}

Assuming that there is no orthogonal quadruple, we show that every pair of vertices of $G$ is at distance at most~4.
For that we repeatedly use that, for every $a,b,c,d \in S$, $\ind(a,b,c,d) := \min\{i \in [\ell] $ $|$ $a[i]=b[i]=c[i]=d[i]=1\}$ is a well-defined index in $[\ell]$.
We only take the minimum index to have a deterministic notation, but there is nothing particular with it, and any index of the non-empty $\{i \in [\ell] $ $|$ $a[i]=b[i]=c[i]=d[i]=1\}$ would work all the same.
%We will also use $\ind(a,b,c)$ as a short-hand for $\ind(a,b,c,c)$.

We first observe that every vertex is at distance at most~3 from $I$.
\begin{lemma}\label{lem:distI}
  $N^1[I] \supseteq I \cup P$, $N^2[I] \supseteq I \cup P \cup C$, and $N^3[I] = V(G)$.
\end{lemma}
\begin{proof}
  The first and second inclusions are actually equalities but we will not need those facts.
  $N^1[I] \supseteq I \cup P$ since every $(\{a,b\},i,j,k) \in P$ is adjacent (with an edge of weight 1) to $(i,i,i,j,k) \in I$.
  Then, $N^2[I] \supseteq N^1[I \cup P] \supseteq I \cup P \cup C$ since every $(a,b,i,j,k) \in C$ is adjacent to $(\{a,a\},i,j,k) \in P$.
  Finally, $N^3[I] \supseteq N^1[I \cup P \cup C] = V(G)$ since every $(a,b,c) \in T$ is adjacent to $(a,b,i,i,i) \in C$ for some $i \in [\ell]$, for otherwise $a,b,c$ is an orthogonal triple.
\end{proof}

We now exhibit paths of length at most~4 between every pair of vertices of $G$.
For the case disjunction, initially imagine the $K_4$ with loops on vertices $T,C,P,I$, where edges correspond to kinds of pairs that are left to check.
The following paragraphs remove all its edges in the order: all edges incident to $I$, all remaining edges incident to $P$ but $TP$, all remaining edges incident to $C$, the loop on $T$, and finally the edge $TP$.

\paragraph*{Between $u \in I$ and $v \in V(G)$}

As $G[I]$ is a clique and, by~\cref{lem:distI}, $N^3[I]=V(G)$, every vertex $u \in I$ is at distance at most 4 from every vertex $v \in V(G)$.

\paragraph*{Between $u \in P$ and $v \in P \cup C$}

For every $u \in P$, $N^2[u] \supset I$ and so $N^4[u] \supset P \cup C$, by~\cref{lem:distI}.
In particular there is a path of length at most 4 between $u$ and any vertex $v \in P \cup C$.

\paragraph*{Between $u \in C$ and $v \in T \cup C$}

Let $(a,b)$ be the two vector fields of $u \in C$, $(c,d)$ be the first two vector fields of $v \in T \cup C$, and $e$ be the third vector field of $v$ if $v \in T$.
Let $i = \ind(a,b,c,d)$, $j = \ind(a,c,d,e)$ if $v \in T$, and $j = i$ if $v \in C$.
We observe that $(a,b,i,i,j), (\{a,c\},i,i,j), (c,d,i,i,j)$ are (existing) vertices of $C$, $P$, and $C$, respectively, and that $u - (a,b,i,i,j) - (\{a,c\},i,i,j) - (c,d,i,i,j)$ is a path of length 3 in $G$.
The existence of these vertices is implied by $a[i] = b[i] = c[i] = d[i] = 1$, $a[j] = c[j] = d[j] = 1$.
The first edge of the path is an index-switching edge within $C$.
The existence of the other edges is implied by $a \in \{a,c\}$, $c \in \{a,c\}$, and the fact that the index tuple $(i,i,j)$ does not change.

Finally if $v \in C$, then the index-switching edge $(c,d,i,i,j) - v$ completes the $uv$-path of length 4.
If instead $v \in T$, then the edge $(c,d,i,i,j) - (c,d,e) = v$ completes the $uv$-path of length 4.
This edge exists since $d[j] = e[j] = 1$.

\paragraph*{Between $u \in T$ and $v \in T$}

Let $u = (a,b,c), v = (d,e,f) \in T$, $i = \ind(a,b,c,d)$, $j = \ind(a,b,d,e)$ and $k = \ind(a,d,e,f)$.
Then $u = (a,b,c) - (a,b,i,j,k) - (\{a,d\},i,j,k) - (d,e,i,j,k) - (d,e,f) = v$ is a path of length 4 in $G$.
These vertices exist since $a$ and $d$ have value 1 on indices $i,j,k$, $b$, on indices $i,j$, and $e$, on indices $j,k$.
The first edge exists since $b[i] = c[i] = 1$, the next two edges exist for similar reasons as invoked in the previous paragraph, and the fourth edge exists since $e[k] = f[k] = 1$. 

\paragraph*{Between $u \in T$ and $v \in P$}

Let $u = (a,b,c) \in T$ and $v = (\{d,e\},i,j,k) \in P$.
We set $p_1 = \ind(a,b,c,d)$, $p_2 = \ind(a,b,c,e)$, and exhibit a $uv$-path of length 4 via $I$.
Indeed $u = (a,b,c) - (p_1,p_2,i,j,k) - (\{d,e\},i,j,k) = v$ is a path of length 4 in $G$ (recall that the first edge of the path has weight~3).
Edge $(a,b,c) - (p_1,p_2,i,j,k) \in E(T,I)$ exists since $a[p_1] = b[p_1] = c[p_1] = a[p_2] = b[p_2] = c[p_2] = 1$.
Edge $(p_1,p_2,i,j,k) - (\{d,e\},i,j,k) \in E(I,P)$ exists since $d[p_1] = e[p_2] = 1$ and the three last indices $(i,j,k)$ remain unchanged. 

\subsection{The presence of orthogonal quadruple implies diameter at least 7}

Let $a,b,c,d \in S$ be an orthogonal quadruple, that is, such that there is \emph{no} index $i \in [\ell]$ satisfying $a[i] = b[i] = c[i] = d[i] = 1$.
We may further assume that $a,b,c,d$ are all distinct since checking for an orthogonal triple can be done in time $\tilde{O}(N^3)$.
We will now show that there is no path $\mathcal P$ of length at most 6 between $u = (a,b,c) \in T$ and $v = (d,c,b) \in T$.

Since the distance between every pair of vertices in $T \times I$ is at least~3, a $TT$-path of length at most~6 cannot contain an edge of the clique $G[I]$, nor more generally intersects $I$ at least twice.
We thus distinguish two cases: (case A) $\mathcal P$ visits $I$ exactly once, and (case B) $\mathcal P$ remains within $T \cup C \cup P$. 
Before proving that no $uv$-path $\mathcal P$ of length at most~6 visits $I$, thereby ruling out case A, we state a couple of useful observations.

\begin{observation}\label{obs:P-to-T}
  There is at most one path of length~2 between $(\{d,e\},i,j,k) \in P$ and $(a,b,c) \in T$, namely $(\{d,e\},i,j,k) - (a,b,i,j,k) - (a,b,c)$, which in particular implies that $a \in \{d,e\}$.
\end{observation}

More basically, the only neighbors of $(a,b,c) \in T$ (at distance 1, so not in $I$) are of the form $(a,b,i,j,k) \in C$.
We can generalize this observation to paths contained in $T \cup C$.

\begin{observation}\label{obs:within-tc}
  For every path within $G[T \cup C]$, all the vertices of the path have the same first two vector fields.    
\end{observation}

\paragraph*{Case A: $\mathcal P$ visiting $I$}

As $\mathcal P$ cannot visit $I$ twice, if it visits $I$ then it has length exactly~6 and is one of the following kinds: 
(case 1) $T - I - T$, (case 2) $T - C - P - I - P - C - T$, or (case 3) $T - I - P - C - T$ (recall that the edges in $E(I,T)$ have weight~3).
An important feature of such paths is that no index-switching edge can be used, thus the three last index fields (when they exist) have to remain the same.

\textbf{Case 1.} A path $(a,b,c) - (p_1,p_2,i,j,k) - (d,c,b)$ would in particular imply that $a[p_1] = b[p_1] = c[p_1] = d[p_1] = 1$, contradicting the orthogonality of $a,b,c,d$.

\textbf{Case 2.} By~\cref{obs:P-to-T} applied to both ends of the path, $\mathcal P$ is of the form $(a,b,c) - (a,b,i,j,k) - (\{a,e\},i,j,k) - (p_1,p_2,i,j,k) - (\{d,f\},i,j,k) - (d,c,i,j,k) - (d,c,b)$ with some $e, f \in S$.
The existence of the vertices $(a,b,i,j,k), (d,c,i,j,k) \in C$ implies that $a[i] = a[j] = a[k] = d[i] = d[j] = d[k] = 1$, and that $b$ and $c$ have value 1 on at least two indices (with multiplicity) of multiset $\{i,j,k\}$.
In particular, there is an $h \in \{i,j,k\}$ such that $a[h] = b[h] = c[h] = d[h] = 1$, a~contradiction to $a,b,c,d$ being orthogonal.

\textbf{Case 3.}
By~\cref{obs:P-to-T} applied to the second half of the path, $\mathcal P$ has then the form $(a,b,c) - (p_1,p_2,i,j,k) - (\{d,e\},i,j,k) - (d,c,i,j,k) - (d,c,b)$.
The first three vertices yield a contradiction.
Indeed, the existence of edge $(p_1,p_2,i,j,k) - (\{d,e\},i,j,k)$ implies that $d[p_z] = 1$ for some $z \in \{1,2\}$, while the existence of $(a,b,c) - (p_1,p_2,i,j,k)$ implies that $a[p_z] = b[p_z] = c[p_z] = 1$.

\paragraph*{Case B: paths $\mathcal P$ within $T \cup C \cup P$}

We now consider paths $\mathcal P$ in $G[T \cup C \cup P]$.
Since $a \neq d$, $\mathcal P$ has to visit $P$, since otherwise the first vector field cannot change, by~\cref{obs:within-tc}. 
We then observe that no shortest $uv$-path visits $T$ a~third time (one more time than the two endpoints $u$ and $v$).
A $TT$-path visiting $T$ a third time would contain a segment $C - T - C$ that can be shortcut into $C - C$.
Indeed, $(a,b,i,j,k) - (a,b,c) - (a,b,i',j',k')$ has a chord $(a,b,i,j,k) - (a,b,i',j',k')$ which is an index-switching edge of $C$.

We further distinguish two cases:
(case 1) $\mathcal P$ does not contain any index-switching edge, or (case 2) $\mathcal P$ contains at least one index-switching edge.

\textbf{Case 1.}
In that case, $\mathcal P$ is of the form $T - C - P - C - T$ or $T - C - P - C - P - C - T$.
Either way, we consider the unique neighbors of $u$ and $v$ in $\mathcal P$.
These neighbors have to be $(a,b,i,j,k) \in C$ and $(d,c,i,j,k) \in C$ for some $i,j,k \in [\ell]$.
Indeed no index-switching edge nor return to $T$ is allowed here.
Thus we conclude as in case~A.2.

\textbf{Case 2.}
We now assume that $\mathcal P$ contains at least one index-switching edge (of $C$).
In that case, as $\mathcal P$ has length at most~6, it can visit $P$ only once.
Hence $\mathcal P$ is of the kind $T - C - C - P - C - C - T$, where one of the two edges $C-C$ is optional.
We consider the last vertex $u' \in C$ before visiting $P$, and the first vertex $v' \in C$ after visiting $P$.
There is, by design, no index-switching edge between $u'$ and $v'$ on path $\mathcal P$.
Thus by~\cref{obs:within-tc}, there are $i,j,k \in [\ell]$ such that $u'=(a,b,i,j,k)$ and $v'=(d,c,i,j,k)$.
We then conclude as in case~A.2.

\section{Removing the weights}\label{sec:unweighted}

So far we showed the announced lower bound for sparse undirected \emph{weighted} \textsc{Diameter}.
We show how to tune the previous construction to get the same lower bound for sparse undirected unweighted \textsc{Diameter}.
The weighted graph $G$ had only non-trivial edge weights in $E(T,I)$.
We now describe how to replace these weighted edges, to get an unweighted graph $G' = \rho'(S)$.

\subsection{Unweighted construction}

We start with a short summary of the changes.
We will replace $T$ by three copies $T, T', T''$ with an induced perfect matching between $T$ and $T'$, and between $T'$ and $T''$.
We link $T''$ to~$I$ as we linked $T$ to~$I$, and $T$ and $C$, and $T'$ and $C$, as we linked $T$ and $C$.
We finally add a~set $I'$ of vertices with empty vector tuple (like $I$) that we link to $T''$ and $I$ only.  

\paragraph*{Addition to the vertex set}

We add three sets to $V(G)$ to get $V(G')$: two identical copies of $T$, denoted by $T'$ and $T''$, and a set $I'$ isomorphic to $[\ell]$.
More precisely, for every $i \in [\ell]$, we add vertex $(i)$ to $I'$.
Thus $I'$ has no vector field and a unique index field.
We use a subscript to distinguish the \emph{homologous} vertices in $T, T', T''$.
Vertices $(a,b,c)_T \in T, (a,b,c)_{T'} \in T', (a,b,c)_{T''} \in {T''}$ are the three vertices of $G'$ corresponding to the same vertex $(a,b,c)$ of $G$. 

\paragraph*{Edition of the edge set}

We first remove the edges of $G$ with weight~3 (between $T$ and $I$).
For every $\{a,b,c\} \in {S \choose 3}$, we add the edges $(a,b,c)_T - (a,b,c)_{T'}$ and $(a,b,c)_{T'} - (a,b,c)_{T''}$.
We also add edges between $T'$ and $C$, the same way we have defined edges between $T$ and $C$. 
That is, $(a,b,c)_{T'} - (a,b,i,j,k)$ is an edge if and only if $(a,b,c)_T - (a,b,i,j,k)$ is an edge.
Let us recall that the existence of this edge (and of its endpoint in $C$) implies that $a, b, c$ have value 1 on indices $\{i,j,k\}$, three times, at least twice, and at least once, respectively, and that there is an $h \in \{i,j,k\}$ such that $a[h] = b[h] = c[h]$.

We add edges (of weight 1) between $T''$ and $I$, the same way we defined the weight-3 edges of $G$ between $T$ and $I$.
Thus there is an edge $(a,b,c)_{T''} - (p_1,p_2,i,j,k)$ whenever $a[p_1] = b[p_1] = c[p_1] = a[p_2] = b[p_2] = c[p_2] = 1$.
We further add an edge between $(i) \in I'$ and $(a,b,c)_{T''}$ whenever $a[i] = b[i] = 1$.
Finally we add all the index-switching edges in $I'$, and we make $I$ and $I'$ fully adjacent, that is, we turn $G'[I' \cup I]$ into a clique.

This finishes the edition to the unweighted construction.
See~\cref{fig:unweighted} for an illustration.

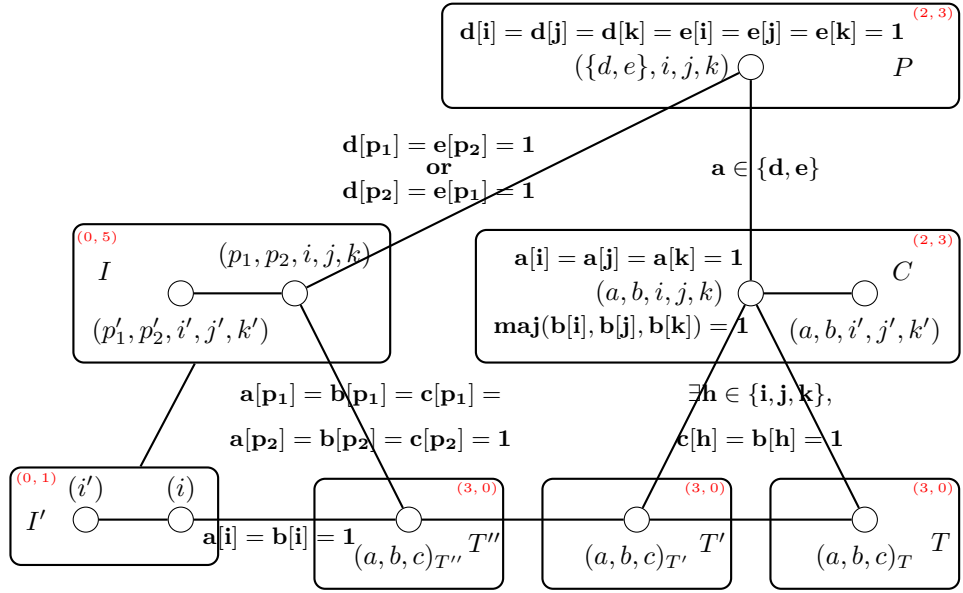
\begin{figure}[h!]
  \centering
  \begin{tikzpicture}
    \def\h{5}
    \def\v{3}
    \foreach \i/\j/\l/\ll/\x/\y in {0.3/0/{(a,b,c)_T}/abc/0/-0.5, -0.3/0/{(a,b,c)_{T'}}/abcp/0/-0.5, -0.9/0/{(a,b,c)_{T''}}/abcpp/0/-0.5, -1.5/0/{(i)}/indi/0/0.4, -1.75/0/{(i')}/indip/0/0.4, 0/1/{(a,b,i,j,k)}/abijk/-1.2/0, 0.3/1/{(a,b,i',j',k')}/abijk2/0/-0.5, 0/2/{(\{d,e\},i,j,k)}/deijk/-1.3/0, -1.2/1/{(p_1,p_2,i,j,k)}/ind/0/0.5,-1.5/1/{(p'_1,p'_2,i',j',k')}/indp/0/-0.5}{
      \node[draw, circle] (\ll) at (\h * \i,\v * \j) {} ;
      \node (t\ll) at (\h * \i + \x,\v * \j + \y) {$\l$} ;
    }

    \foreach \i/\j/\c in {abc/abijk/black, abcp/abijk/black, abc/abcp/black, abcp/abcpp/black, abcpp/ind/black, abcpp/indi/black, indi/indip/black, abijk/abijk2/black,abijk/deijk/black,ind/deijk/black,ind/indp/black}{
      \draw[thick, color=\c] (\i) -- (\j) ;
    }
    \node (d1) at (0.1 * \h, 0.1 * \v) {} ;
    \node (d2) at (0.5 * \h, -0.1 * \v) {} ;
    \node (e1) at (-0.5 * \h, 0.1 * \v) {} ;
    \node (e2) at (-0.1 * \h, -0.1 * \v) {} ;
    \node (f1) at (-1.1 * \h, 0.1 * \v) {} ;
    \node (f2) at (-0.7 * \h, -0.1 * \v) {} ;
    \node (g1) at (-1.45 * \h, 0.15 * \v) {} ;
    \node (g2) at (-1.9 * \h, -0.12 * \v) {} ;
    \node (d1p) at (- 0.5 * \h, 2.1 * \v) {} ;
    \node (d2p) at (0.5 * \h, 1.9 * \v) {} ;

    \small{}
    \node (ab1) at (- 0.32 * \h, 1.15 * \v) {$\mathbf{a[i]=a[j]=a[k]=1}$} ;
    \node (ab2) at (- 0.34 * \h, 0.85 * \v) {\textbf{maj}$\mathbf{(b[i],b[j],b[k])=1}$} ;
    \node (de1) at (- 0.18 * \h, 2.15 * \v) {$\mathbf{d[i]=d[j]=d[k]=e[i]=e[j]=e[k]=1}$} ;
    \node at (0.025 * \h, 0.55 * \v) {$\mathbf{\exists h \in \{i,j,k\},}$} ;
    \node at (0.025 * \h, 0.35 * \v) {$\mathbf{c[h]=b[h]=1}$} ;
    \node at (0.04 * \h, 1.55 * \v) {$\mathbf{a \in \{d,e\}}$} ;
    \node at (-1 * \h, 0.55 * \v) {$\mathbf{a[p_1]=b[p_1]=c[p_1]=}$} ;
    \node at (-1 * \h, 0.35 * \v) {$\mathbf{a[p_2]=b[p_2]=c[p_2]=1}$} ;
    \node at (-1.24 * \h, -0.08 * \v) {$\mathbf{a[i]=b[i]=1}$} ;
    \node at (-0.82 * \h, 1.65 * \v) {$\mathbf{d[p_1]=e[p_2]=1}$} ;
    \node at (-0.82 * \h, 1.55 * \v) {\textbf{or}} ;
    \node at (-0.82 * \h, 1.45 * \v) {$\mathbf{d[p_2]=e[p_1]=1}$} ;
    \normalsize{}

    \node[draw,rectangle,rounded corners,thick,fit=(ab1) (ab2) (tabijk2)] (C) {} ;
    \node[draw,rectangle,rounded corners,thick,fit=(abc) (tabc) (d1) (d2)] (T) {} ;
    \node[draw,rectangle,rounded corners,thick,fit=(abcp) (tabcp) (e1) (e2)] (Tp) {} ;
    \node[draw,rectangle,rounded corners,thick,fit=(abcpp) (tabcpp) (f1) (f2)] (Tpp) {} ;
    \node[draw,rectangle,rounded corners,thick,fit=(de1)  (deijk) (tdeijk) (d1p) (d2p)] (P) {} ;
    \node[draw,rectangle,rounded corners,thick,fit=(tind) (tindp)] (I) {} ;
    \node[draw,rectangle,rounded corners,thick,fit=(g1) (g2)] (Ip) {} ;

    \draw[thick] (I) -- (Ip) ;

    \foreach \i/\j/\k in {0.5/-0.1/T, -0.1/-0.1/T', -0.7/-0.1/T'',-1.88/0/I', 0.4/1.1/C, 0.4/2/P, -1.7/1.1/I}{
      \node (t\k) at (\i * \h, \j * \v) {$\k$} ;
    }
    \foreach \i/\j/\k in {0.48/0.135/{(3,0)}, -0.12/0.135/{(3,0)}, -0.72/0.135/{(3,0)}, -1.88/0.18/{(0,1)}, 0.48/1.235/{(2,3)}, 0.48/2.235/{(2,3)}, -1.72/1.245/{(0,5)}}{
      \node[red] at (\i * \h, \j * \v) {\tiny{$\k$}} ;
    }
  \end{tikzpicture}
  \caption{The unweighted construction $G'$.
  In bold, the conditions for the existence of a vertex or of an edge.
  The pairs in red recall, for vertices of the corresponding set, the length of their vector tuple followed by the length of their index tuple.}
  \label{fig:unweighted}
\end{figure}

\paragraph*{New vertex and edge count}

We added to $V(G)$ $O(N^3)$ vertices in $T' \cup T''$, and $\ell$, in $I'$.
Thus $G'$ has also $O(|V(G)|) = O(N^3+N^2 \ell^3 + \ell^5)=O(N^3)$ vertices. 
We added to $E(G)$ $O(N^3+N^3 \ell^3)$ edges incident to~$T'$, and $O(N^3 \ell+\ell^6+\ell^2)$ edges incident to $I'$.
(The edges between $T''$ and $I$ were already counted in $G$ between $T$ and $I$.)
Thus $G'$ has $O(|E(G)|) = O(N^3 \ell^5+N^2 \ell^6 + \ell^{10})=\tilde{O}(N^3)$ edges.
Again $G'$ can be computed in time $\tilde{O}(N^3)$.

\subsection{The absence of orthogonal quadruple implies diameter at most 4}

In case $S$ has no orthogonal quadruple, we use similar arguments as in $G$, to find paths of length at most~4 between every pair of vertices in $G'$.
We first show that $I'$ is at distance at most~3 of every vertex of $G'$.
\begin{lemma}\label{lem:distIp}
  $N[I'] \supseteq I' \cup I \cup T''$, $N^2[I'] \supseteq I' \cup I \cup T'' \cup P \cup T'$, and $N^3[I']=V(G')$.
\end{lemma}
\begin{proof}
  The inclusions are actually equalities.
  $N[I'] \supseteq I' \cup I \cup T''$ since $I$ is fully adjacent to~$I'$ and every $(a,b,c)_{T''} \in T''$ is adjacent to some $(i) \in I'$, for otherwise $a,b$ is an orthogonal pair.
  $N^2[I'] \supseteq N[I' \cup I \cup T''] \supseteq I' \cup I \cup T'' \cup P \cup T'$ since every $(\{a,b\},i,j,k) \in P$ is adjacent to $(i,i,i,j,k) \in I$ and every $(a,b,c)_{T'} \in T'$ is adjacent to $(a,b,c)_{T''} \in T''$.
  Finally, $N^3[I'] \supseteq N[I' \cup I \cup T'' \cup P \cup T'] = V(G')$ since every $(a,b,i,j,k) \in C$ is adjacent to $(\{a,a\},i,j,k) \in P$ and every $(a,b,c)_T \in T$ is adjacent to $(a,b,c)_{T'} \in T'$.
\end{proof}

We also show the following inclusions.
\begin{lemma}\label{lem:distI2}
  $N[I] \supset P \cup T''$, and $N^2[I] \supset P \cup C \cup T'' \cup T'$.
\end{lemma}
\begin{proof}
  $N[I] \supset P$, $N^2[I] \supset C$, and $N[T''] \supset T'$ have all been shown in~\cref{lem:distIp}.
  Therefore we shall just prove that $N[I] \supset T''$.
  Indeed every vertex $(a,b,c)_{T''} \in T''$ is adjacent to some $(i,i,i,i,i) \in I$, since otherwise $a,b,c$ is an orthogonal triple.
\end{proof}

For the case disjunction, initially imagine the $K_7$ with loops on vertices $T,T',T'',C,P,I,I'$, where edges correspond to the kinds of pairs that are left to check.
The following paragraphs remove all its edges in the order: all edges incident to $I$ and to $I'$, all remaining edges incident to $P$ and to $T''$ but $TP$ and $TT''$, all remaining edges incident to $C$, all remaining edges incident to $T'$ as well the loop on $T$, the edge $TP$, and finally the edge $TT''$.

\paragraph*{Between $u \in I \cup I'$ and $v \in V(G')$}

As $G'[I \cup I']$ is a clique and, by \cref{lem:distIp}, $N^3[I'] = V(G')$, then $N^4[u] = V(G')$ holds for every vertex $u \in I \cup I'$.

\paragraph*{Between $u \in P \cup T''$ and $v \in P \cup C \cup T'' \cup T'$}

For every $u \in P \cup T''$, by~\cref{lem:distI2} and the fact that $G'[I]$ is a clique, $N^2[u] \supset I$ and, again by~\cref{lem:distI2}, $N^4[u] \supset P \cup C \cup T'' \cup T'$.
In particular there is a path of length at most 4 from $u$ to any vertex $v \in P \cup C \cup T'' \cup T'$.

The following two cases work as in $G$, since $(a,b,c)_T$ and $(a,b,c)_{T'}$ are twins in $G'[T \cup T' \cup C \cup P]$.

\paragraph*{Between $u \in C$ and $v \in T \cup T' \cup C$}

This holds by replacing the occurrence of $(c,d,e)$ by $(c,d,e)_T$ or $(c,d,e)_{T'}$, and every occurrence of $T$ by $T \cup T'$, in the paragraph \emph{Between $u \in C$ and $v \in T \cup C$} of the weighted construction.

\paragraph*{Between $u \in T \cup T'$ and $v \in T \cup T'$}

Again this holds by replacing occurrences of $(a,b,c)$ (resp.~$(c,d,e)$) by $(a,b,c)_T$ or $(a,b,c)_{T'}$ (resp.~$(c,d,e)_T$ or $(c,d,e)_{T'}$).

\paragraph*{Between $u \in T$ and $v \in P$}

This works as in $G$ by following three edges of weight~1 from $T$ to $I$, instead of a single edge of weight~3.
For every $u = (a,b,c)_T \in T$ and $v = (\{d,e\},i,j,k) \in P$, there is a path $u = (a,b,c)_T - (a,b,c)_{T'} - (a,b,c)_{T''} - (p_1,p_2,i,j,k) - (\{d,e\},i,j,k) = v$ in $G'$, with $p_1 = \ind(a,b,c,d)$ and $p_2 = \ind(a,b,c,e)$.

\paragraph*{Between $u \in T$ and $v \in T''$}

This case is the real novelty compared to $G$, and the reason for introducing $I'$. 
For every $u = (a,b,c)_T \in T$ and $v = (d,e,f)_{T''} \in T''$, there is a path $u = (a,b,c)_T - (a,b,c)_{T'} - (a,b,c)_{T''} - (i) - (d,e,f)_{T''} = v$ in $G'$, with $i = \ind(a,b,d,e)$.
The last two edges exist since $a[i] = b[i] = d[i] = e[i] = 1$.

\subsection{The presence of orthogonal quadruple implies diameter at least 7}

Again we assume that there is an orthogonal quadruple $a,b,c,d \in S$ such that $a,b,c,d$ are pairwise distinct.
We claim that there is no path of length at most~6 in $G'$ between $u = (a,b,c)_T$ and $v = (d,c,b)_T$.
Since the distance between $T$ and $I \cup I'$ is at least~3, any $TT$-path of length at most~6 visits $I \cup I'$ at most once.
For the sake of contradiction, let $\mathcal P$ be such a path that we further assume shortest (hence in particular chordless) and, among shortest $uv$-paths, having the fewest edges in $E(T',C)$.
We will show that $\mathcal P$ cannot visit $I'$, nor use any edge of $E(T',C)$.
Finally we observe that $TT$-paths of length at most~6 in $G'$ respecting these two interdictions are in length-preserving one-to-one correspondence with $TT$-paths in $G$. 

\paragraph*{$\mathcal P$ cannot visit $I'$}

The only possible kind of a $TT$-path of length at most~6 visiting $I'$ is $T - T' - T'' - I' - T'' - T' - T$.
This forces $\mathcal P$ to be of the form $(a,b,c)_T - (a,b,c)_{T'} - (a,b,c)_{T''} - (i) - (d,c,b)_{T''} - (d,c,b)_{T'} - (d,c,b)_T$ for some $i \in [\ell]$.
However the third and fourth edges imply that there is an $i \in [\ell]$ such that $a[i] = b[i] = c[i] = d[i] = 1$, a contradiction to the orthogonality of $a,b,c,d$.

\paragraph*{$\mathcal P$ cannot use any edge of $E(T',C)$}

Assuming that $\mathcal P$ contains at least one edge in $E(T',C)$, we first show that it has to contain a subpath $C - T' - T'' - I \cup I'$ or $I \cup I' - T'' - T' - C$.
Let $w = (a',b',c')_{T'} \in T' \cap V(\mathcal P)$ be a vertex of $\mathcal P$ with one neighbor $x \in C \cap V(\mathcal P)$ on $\mathcal P$.
The other neighbor $y$ of $w$ on $\mathcal P$ is necessarily in $T''$.
Indeed if $y \in T$, then $y = (a',b',c')_T$, and $xy \in E(G')$ is a chord.
If instead $y \in C$, then one can replace the subpath $x - (a',b',c')_{T'} - y$ by $x - (a',b',c')_T - y$, contradicting the minimality of the number of used edges in $E(T',C)$ (since this number decreases by~2).

Thus the only possibility is that $y \in T''$.
Then the other neighbor of $y$ on $\mathcal P$ (other than~$w$) has to be in $I \cup I'$, since otherwise $\mathcal P$ is not a \emph{simple} path.
Hence $\mathcal P$ contains a subpath of the kind $C - T' - T'' - I \cup I'$ (or the reverse, $I \cup I' - T'' - T' - C$).
Now we observe that $C$ is at distance at least 1 from $T$, while $I \cup I'$ is at distance at least~3 from $T$.
Therefore such a path $\mathcal P$ would have length at least~7.

\paragraph*{Such a path $\mathcal P$ would also exist in $G$}

We can now assume that $\mathcal P$ does not use any vertex of $I'$ nor any edge of $E(T',C)$.
Every such \emph{simple} $TT$-path (visiting $I$ at most once) also \emph{exists} in the weighted graph $G$, with the same length.
To see it, we notice that if $\mathcal P$ contains an edge $(a',b',c')_T - (a',b',c')_{T'}$, then it has to contain a subpath of the form $(a',b',c')_T - (a',b',c')_{T'} - (a',b',c')_{T''} - (p_1,p_2,i,j,k) \in I$, and is emulated in $G$ by taking the weight-3 edge $(a',b',c') - (p_1,p_2,i,j,k)$.
However we showed in the previous section that no $uv$-path of length at most~6 exists in $G$.

%\bibliography{main}

\begin{thebibliography}{10}

\bibitem{Aingworth99}
Donald Aingworth, Chandra Chekuri, Piotr Indyk, and Rajeev Motwani.
\newblock {Fast Estimation of Diameter and Shortest Paths (Without Matrix
  Multiplication)}.
\newblock {\em {SIAM} J. Comput.}, 28(4):1167--1181, 1999.
\newblock URL: \url{https://doi.org/10.1137/S0097539796303421}, \href
  {http://dx.doi.org/10.1137/S0097539796303421}
  {\path{doi:10.1137/S0097539796303421}}.

\bibitem{Backurs18}
Arturs Backurs, Liam Roditty, Gilad Segal, Virginia~Vassilevska Williams, and
  Nicole Wein.
\newblock Towards tight approximation bounds for graph diameter and
  eccentricities.
\newblock In Ilias Diakonikolas, David Kempe, and Monika Henzinger, editors,
  {\em Proceedings of the 50th Annual {ACM} {SIGACT} Symposium on Theory of
  Computing, {STOC} 2018, Los Angeles, CA, USA, June 25-29, 2018}, pages
  267--280. {ACM}, 2018.
\newblock URL: \url{https://doi.org/10.1145/3188745.3188950}, \href
  {http://dx.doi.org/10.1145/3188745.3188950}
  {\path{doi:10.1145/3188745.3188950}}.

\bibitem{Bonnet20}
{\'{E}}douard Bonnet.
\newblock {Inapproximability of Diameter in super-linear time: Beyond the 5/3
  ratio}.
\newblock {\em CoRR, To appear at STACS 2021}, abs/2008.11315, 2020.
\newblock URL: \url{https://arxiv.org/abs/2008.11315}, \href
  {http://arxiv.org/abs/2008.11315} {\path{arXiv:2008.11315}}.

\bibitem{Cairo16}
Massimo Cairo, Roberto Grossi, and Romeo Rizzi.
\newblock {New Bounds for Approximating Extremal Distances in Undirected
  Graphs}.
\newblock In Robert Krauthgamer, editor, {\em Proceedings of the Twenty-Seventh
  Annual {ACM-SIAM} Symposium on Discrete Algorithms, {SODA} 2016, Arlington,
  VA, USA, January 10-12, 2016}, pages 363--376. {SIAM}, 2016.
\newblock URL: \url{https://doi.org/10.1137/1.9781611974331.ch27}, \href
  {http://dx.doi.org/10.1137/1.9781611974331.ch27}
  {\path{doi:10.1137/1.9781611974331.ch27}}.

\bibitem{Carmosino16}
Marco~L. Carmosino, Jiawei Gao, Russell Impagliazzo, Ivan Mihajlin, Ramamohan
  Paturi, and Stefan Schneider.
\newblock Nondeterministic extensions of the strong exponential time hypothesis
  and consequences for non-reducibility.
\newblock In Madhu Sudan, editor, {\em Proceedings of the 2016 {ACM} Conference
  on Innovations in Theoretical Computer Science, Cambridge, MA, USA, January
  14-16, 2016}, pages 261--270. {ACM}, 2016.
\newblock URL: \url{https://doi.org/10.1145/2840728.2840746}, \href
  {http://dx.doi.org/10.1145/2840728.2840746}
  {\path{doi:10.1145/2840728.2840746}}.

\bibitem{Chechik14}
Shiri Chechik, Daniel~H. Larkin, Liam Roditty, Grant Schoenebeck, Robert~Endre
  Tarjan, and Virginia~Vassilevska Williams.
\newblock {Better Approximation Algorithms for the Graph Diameter}.
\newblock In Chandra Chekuri, editor, {\em Proceedings of the Twenty-Fifth
  Annual {ACM-SIAM} Symposium on Discrete Algorithms, {SODA} 2014, Portland,
  Oregon, USA, January 5-7, 2014}, pages 1041--1052. {SIAM}, 2014.
\newblock URL: \url{https://doi.org/10.1137/1.9781611973402.78}, \href
  {http://dx.doi.org/10.1137/1.9781611973402.78}
  {\path{doi:10.1137/1.9781611973402.78}}.

\bibitem{Cygan16}
Marek Cygan, Holger Dell, Daniel Lokshtanov, D{\'{a}}niel Marx, Jesper
  Nederlof, Yoshio Okamoto, Ramamohan Paturi, Saket Saurabh, and Magnus
  Wahlstr{\"{o}}m.
\newblock {On Problems as Hard as CNF-SAT}.
\newblock {\em {ACM} Trans. Algorithms}, 12(3):41:1--41:24, 2016.
\newblock URL: \url{https://doi.org/10.1145/2925416}, \href
  {http://dx.doi.org/10.1145/2925416} {\path{doi:10.1145/2925416}}.

\bibitem{Cygan15}
Marek Cygan, Fedor~V. Fomin, Lukasz Kowalik, Daniel Lokshtanov, D{\'{a}}niel
  Marx, Marcin Pilipczuk, Michal Pilipczuk, and Saket Saurabh.
\newblock {\em {Parameterized Algorithms}}.
\newblock Springer, 2015.
\newblock URL: \url{https://doi.org/10.1007/978-3-319-21275-3}, \href
  {http://dx.doi.org/10.1007/978-3-319-21275-3}
  {\path{doi:10.1007/978-3-319-21275-3}}.

\bibitem{Wein20}
Mina Dalirrooyfard and Nicole Wein.
\newblock {Tight Conditional Lower Bounds for Approximating Diameter in
  Directed Graphs}.
\newblock {\em CoRR}, abs/2011.03892, 2020.
\newblock URL: \url{https://arxiv.org/abs/2011.03892}, \href
  {http://arxiv.org/abs/2011.03892} {\path{arXiv:2011.03892}}.

\bibitem{Sparsification}
Russell Impagliazzo and Ramamohan Paturi.
\newblock On the complexity of k-sat.
\newblock {\em J. Comput. Syst. Sci.}, 62(2):367--375, 2001.
\newblock URL: \url{https://doi.org/10.1006/jcss.2000.1727}, \href
  {http://dx.doi.org/10.1006/jcss.2000.1727}
  {\path{doi:10.1006/jcss.2000.1727}}.

\bibitem{Impagliazzo01}
Russell Impagliazzo, Ramamohan Paturi, and Francis Zane.
\newblock {Which Problems Have Strongly Exponential Complexity?}
\newblock {\em J. Comput. Syst. Sci.}, 63(4):512--530, 2001.
\newblock URL: \url{https://doi.org/10.1006/jcss.2001.1774}, \href
  {http://dx.doi.org/10.1006/jcss.2001.1774}
  {\path{doi:10.1006/jcss.2001.1774}}.

\bibitem{Li20}
Ray Li.
\newblock {Improved SETH-hardness of unweighted Diameter}.
\newblock {\em CoRR}, abs/2008.05106, 2020.
\newblock URL: \url{https://arxiv.org/abs/2008.05106}, \href
  {http://arxiv.org/abs/2008.05106} {\path{arXiv:2008.05106}}.

\bibitem{Li20b}
Ray Li.
\newblock {Settling SETH vs. Approximate Sparse Directed Unweighted Diameter
  (up to (NU)NSETH)}.
\newblock {\em CoRR}, abs/2008.05106, 2020.
\newblock URL: \url{https://arxiv.org/abs/2008.05106}, \href
  {http://arxiv.org/abs/2008.05106} {\path{arXiv:2008.05106}}.

\bibitem{Lokshtanov11}
Daniel Lokshtanov, D{\'{a}}niel Marx, and Saket Saurabh.
\newblock {Lower bounds based on the Exponential Time Hypothesis}.
\newblock {\em Bull. {EATCS}}, 105:41--72, 2011.
\newblock URL: \url{http://eatcs.org/beatcs/index.php/beatcs/article/view/92}.

\bibitem{Roditty13}
Liam Roditty and Virginia~Vassilevska Williams.
\newblock Fast approximation algorithms for the diameter and radius of sparse
  graphs.
\newblock In Dan Boneh, Tim Roughgarden, and Joan Feigenbaum, editors, {\em
  Symposium on Theory of Computing Conference, STOC'13, Palo Alto, CA, USA,
  June 1-4, 2013}, pages 515--524. {ACM}, 2013.
\newblock URL: \url{https://doi.org/10.1145/2488608.2488673}, \href
  {http://dx.doi.org/10.1145/2488608.2488673}
  {\path{doi:10.1145/2488608.2488673}}.

\bibitem{Rubinstein19}
Aviad Rubinstein and Virginia~Vassilevska Williams.
\newblock {SETH vs Approximation}.
\newblock {\em {SIGACT} News}, 50(4):57--76, 2019.
\newblock URL: \url{https://doi.org/10.1145/3374857.3374870}, \href
  {http://dx.doi.org/10.1145/3374857.3374870}
  {\path{doi:10.1145/3374857.3374870}}.

\bibitem{Williams05}
Ryan Williams.
\newblock A new algorithm for optimal 2-constraint satisfaction and its
  implications.
\newblock {\em Theor. Comput. Sci.}, 348(2-3):357--365, 2005.
\newblock URL: \url{https://doi.org/10.1016/j.tcs.2005.09.023}, \href
  {http://dx.doi.org/10.1016/j.tcs.2005.09.023}
  {\path{doi:10.1016/j.tcs.2005.09.023}}.

\bibitem{WilliamsSurvey}
Virginia~Vassilevska Williams.
\newblock On some fine-grained questions in algorithms and complexity.
\newblock In {\em Proceedings of the ICM}, volume~3, pages 3431--3472. World
  Scientific, 2018.

\end{thebibliography}

\end{document}